\documentclass[letterpaper, 10 pt, journal]{ieeeconf}
\IEEEoverridecommandlockouts

\usepackage{amsfonts,amsmath,amssymb} 

\def\rbb{\mathbb{R}}
\def\trp{^T}

\def\tr{\mathop{\rm Tr}\nolimits} 

\def\half{\frac{1}{2}}
\def\re{{\mathbb R}}
\def\C{{\mathbb C}} 
\def\begce{\begin{center}}
\def\endce{\end{center}}
\def\begar{\begin{array}}
\def\endar{\end{array}}
\def\begeq{\begin{equation}}
\def\endeq{\end{equation}}
\def\begdi{\begin{displaymath}}
\def\enddi{\end{displaymath}}
\def\begdis{\begin{eqnarray*}}
\def\enddis{\end{eqnarray*}}
\def\begeqa{\begin{eqnarray}}
\def\endeqa{\end{eqnarray}} 
\usepackage{psfrag,color}
\usepackage{enumerate,cite,latexsym,graphicx}
\newtheorem{theorem}{Theorem}

\def\tr{\mathop{\rm Tr}\nolimits} 

\title{\LARGE \bf  A Direct Coupling Coherent Quantum Observer for a Qubit, including Observer Measurements}

\author{Ian R.~Petersen and Elanor H. Huntington
\thanks{This work was supported by the
Australian Research Council (ARC) under grant FL110100020 and the Air Force Office of Scientific
Research (AFOSR), under agreement number FA2386-16-1-4065. }%
\thanks{Ian R. Petersen is with the Research School of  Engineering, The Australian National University, Canberra, ACT 2601, Australia.
         {\tt\small i.r.petersen@gmail.com} } 
\thanks{Elanor H. Huntington is with the 
Research School of Engineering, The Australian National University, Canberra, ACT 0200,
Australia. Email: Elanor.Huntington@anu.edu.au.}
}%

\begin{document}

\maketitle
\thispagestyle{empty}
\pagestyle{empty}

\begin{abstract}
This paper proposes a direct coupling coherent quantum observer for a quantum plant which consists of a two level quantum system. The quantum observer, which is a quantum harmonic oscillator, includes homodyne detection measurements. It is shown that the observer can be designed so that it does not affect the quantum variable of interest in the quantum plant and that measured output converges in a given sense to the plant variable of interest. Also, the plant variable of interest-observer system can be described by a set of linear quantum stochastic differential equations. A minimum variance unbiased estimator form of the Kalman filter is derived for linear quantum systems and applied to the direct coupled coherent quantum observer. 
\end{abstract}

\section{Introduction} \label{sec:intro}
 A number of papers have recently considered the problem of constructing a coherent quantum observer for a quantum system; e.g., see \cite{VP9a,MEPUJ1a,MJP1}. In the coherent quantum observer problem, a quantum plant is coupled to a quantum observer which is also a quantum system. The quantum observer is constructed to be a physically realizable quantum system  so that the system variables of the quantum observer converge in some suitable sense to the variables of interest for the quantum plant. 

The papers \cite{PET14Aa,PET14Ba,PET14Ca,PET14Da}  considered the problem of constructing a direct coupling quantum observer for a given quantum system. In particular, the paper \cite{PET14Ba} considered the case in which the quantum plant was a two level system and the quantum observer was a linear quantum harmonic oscillator. Also, the papers \cite{PeHun1a,PeHun3a} considered the problem of whether such direct coupling coherent observers could be experimentally implemented. In addition, the paper \cite{PeHun2a} considered with the direct coupling coherent observer of \cite{PET14Aa} could be experimentally implemented in an experiment which included homodyne detection measurements of the observer. 

In this paper, we build on the results of \cite{PET14Ba} and \cite{PeHun2a} to consider the case in which the quantum plant is a  two level system and the quantum observer is a linear quantum harmonic oscillator subject to measurements using homodyne detection. Similar convergence results for the quantum observer as obtained in \cite{PeHun2a} are obtained in this case. Also, as in \cite{PET14Ba}, the plant observer system considering only the plant variable of interest is described by a set of linear quantum stochastic differential equations (QSDEs) in spite of the fact that finite level systems are normally described in terms of bilinear QSDEs; e.g., see \cite{EMPUJ4}. However, in this case, measurements are available from the quantum observer. This means that we can apply a version of the Kalman filter to the linear  plant observer QSDEs. 

The paper develops a notion of a minimum variance unbiased estimator for a general set of linear QSDEs,  building on the fact that the classical Kalman Filter can be regarded as the minimum variance unbiased estimator even in the case of non-Gaussian noises and initial conditions; e.g., see \cite{AT67,KS72}. The equations for this estimator are developed for the general case and then applied to the particular case of the plant observer system. This provides a numerically straightforward way of estimating the variable of interest for the qubit system when using homodyne detection measurements.

\section{Direct Coupling Coherent Quantum Observer with Observer Measurement}
We first consider  the dynamics of a single qubit spin system, which will correspond to the quantum plant; see also \cite{EMPUJ4}.  
 The quantum mechanical behavior of the system is described in terms of the system \emph{observables} which are self-adjoint operators on the complex Hilbert space $\mathfrak{H}_p = \C^2$.   The commutator of two scalar operators $x$ and $y$ in ${\mathfrak{H}_p}$ is  defined as $[x, y] = xy - yx$.~Also, for a  vector of operators $x$ in ${\mathfrak H}_p$, the commutator of ${x}$ and a scalar operator $y$ in ${\mathfrak{H}_p}$ is the  vector of operators $[{x},y] = {x} y - y {x}$, and the commutator of ${x}$ and its adjoint ${x}^\dagger$ is the  matrix of operators 
\begdi [{x},{x}^\dagger] \triangleq {x} {x}^\dagger - ({x}^\# {x}^T)^T, \enddi 
where ${x}^\# \triangleq (x_1^\ast\; x_2^\ast \;\cdots\; x_n^\ast )^T$ and $^\ast$ denotes the operator adjoint. In the case of complex vectors (matrices) $^\ast$ denotes the complex conjugate while $^\dagger$ denotes the conjugate transpose. 

The vector of system variables for the single qubit spin system under consideration is 
\begdi x_p=(x_1,x_2,x_3)^T\triangleq (\sigma_1,\sigma_2,\sigma_3),\enddi 
where $\sigma_1$, $\sigma_2$ and $\sigma_3$ are spin operators. Here, $x_p$ is a self-adjoint vector of operators; i.e., $x_p=x_p^\#$.~In particular $x_p(0)$ is represented by the Pauli matrices; i.e.,
\begin{eqnarray*}
\sigma_1(0)&=&\left(\begin{array}{cc}
         0 & 1 \\ 1 & 0
        \end{array} \right),\;\; 
\sigma_2(0)=\left(\begin{array}{cc}
         0 & -{\pmb i} \\ {\pmb i} & 0
        \end{array} \right),\\
\sigma_3(0)&=&\left(\begin{array}{cc}
         1 & 0 \\ 0 & -1
        \end{array} \right).
\end{eqnarray*}
 Products of the spin operators satisfy 
\begeq 
\label{Pauli_prod}
\sigma_i\sigma_j = \delta_{ij}+ {\pmb i} \sum_{k}\epsilon_{ijk}\sigma_k.
\endeq
 It  then follows that the commutation relations for the spin operators are
\begeq \label{eq:Pauli_CCR}
[\sigma_i,\sigma_j] = 2{\pmb i} \sum_{k}\epsilon_{ijk}\sigma_k,
\endeq
where $\delta_{ij}$ is the Kronecker delta and $\epsilon_{ijk}$ denotes the Levi-Civita tensor. The dynamics of the system variables $x$ are determined by the system Hamiltonian which is a self-adjoint operator on $\mathfrak{H}_p$. The {Hamiltonian} is chosen to be linear in $x_p$; i.e., 
\begdi {\mathcal{H}}_p=r_p^T x_p(0) \;\;\enddi
where $r_p\in \re^3$. 
The plant model is then given by the differential equation
\begin{eqnarray}
\dot x_p(t) &=& -{\pmb i}[x_p(t),\mathcal{H}_p]; \nonumber \\
&=& A_px_p(t); \quad x_p(0)=x_{0p}; \nonumber \\
z_p(t) &=& C_px_p(t)
 \label{plant}
\end{eqnarray}
where $z_p$ denotes the  system variable to be estimated by the observer and  $C_p\in \rbb^{1\times 3}$; e.g., see \cite{EMPUJ4}. Also, $A_p \in \re^{3\times 3}$. In order to obtain an expression for the matrix $A_p$ in terms of $r_p$, we  define the linear mapping
$\Theta: \C^3 \rightarrow \C^{3\times 3}$ as
\begeq \label{eq:Theta_definition}
\Theta(\beta)= \left(\begin{array}{ccc}
         0 & \beta_3 & -\beta_2 \\ -\beta_3 & 0 & \beta_1 \\ \beta_2  & -\beta_1  & 0
        \end{array} \right).
\endeq  
Then, it was shown in \cite{EMPUJ4} that
\begdi x_p(t)x_p(t)^T=I + {\pmb i} \Theta(x_p(t)). \enddi 
Similarly, the commutation relations for the spin operators are written as 
\begin{equation}
\label{comm_1}
 [x_p(t),x_p(t)^T]=2 {\pmb i} \Theta(x_p(t)). 
\end{equation}
Also, it was shown in \cite{EMPUJ4} that 
\begin{equation}
-{\pmb i}[x_p(t),r_p^T x_p(t)] = - 2  \Theta(r_p) x_p(t)
\label{Ap}
\end{equation}
and hence $A_p = - 2  \Theta(r_p)$. 

In addition, it is shown in \cite{EMPUJ4} that the mapping $\Theta(\cdot)$ has the following properties:
\begin{eqnarray}
\label{eq:Theta_1} \Theta(\beta)\gamma &=& - \Theta(\gamma) \beta,\\
 \label{eq:Theta_beta_beta} \Theta(\beta)\beta &=& 0,\\
 \label{eq:Theta_multiplication} \Theta(\beta)\Theta(\gamma) &=& \gamma \beta^T -\beta^T \gamma I,\\
 \label{eq:Theta_composition} \Theta\left(\Theta(\beta)\gamma\right)&=&\Theta(\beta)\Theta(\gamma) - \Theta(\gamma)\Theta(\beta).
\end{eqnarray}
Note that a quantum system of this form will be physically realizable which means that the commutation relation (\ref{comm_1}) will hold for all times $t \geq 0$. 

We now describe the linear quantum system  which will correspond to the quantum observer; see also \cite{JNP1,PET10B,NJD09,GJ09,ZJ11}. 
This system is described by QSDEs of the form
\begin{eqnarray}
dx_o &=& A_ox_odt+B_odw;\quad x_o(0)=x_{0o};\nonumber \\
dy_o &=& C_ox_odt+dw;\nonumber \\
z_o &=& Ky_o
 \label{observer}
\end{eqnarray}
where $dw = \left[\begin{array}{l}dQ\\dP\end{array}\right]$ is a $2\times 1$ vector of quantum noises expressed in quadrature form corresponding to the input field for the observer and $dy_o$ is the corresponding output field; e.g., see \cite{JNP1,NJD09}.  The observer output $z_o$ will be  a real scalar quantity obtained by applying homodyne detection to the observer output field. $A_o \in \rbb^{2\times 2}$, $B_o \in \rbb^{2\times 2}$, $C_o\in \rbb^{2 \times 2}$.  Also,  $x_o=\left[\begin{array}{l}q_o\\p_o\end{array}\right]$  is a vector of self-adjoint 
system variables corresponding to the observer position and momentum operators; e.g., see \cite{JNP1}. We  assume that the plant variables commute with the observer variables. The system dynamics (\ref{observer}) are determined by the observer system Hamiltonian and coupling operators which are operators on the underlying  Hilbert space for the observer. For the quantum observer under consideration, this Hamiltonian is a self-adjoint operator given by the 
quadratic form:
$\mathcal{H}_o=\half x_o(0)\trp R_o x_o(0)$, where $R_o$ is a real symmetric matrix. Also, the coupling operator $L$ is defined by a matrix $W_o \in \rbb^{2\times 2}$ so that
\begin{equation}
\label{Wo}
\left[\begin{array}{c} L+ L^* \\ \frac{ L - L^*}{i}\end{array}\right]
 = W_o x_o.
\end{equation}
Then, the corresponding matrices $A_o$, $B_o$ and $C_o$ in 
(\ref{observer}) are given by 
\begin{equation}
\label{eq_coef_cond_Ao}
A_o=2J R_o+\frac{1}{2}W_o^TJW_o, ~~B_o = JW_o^TJ,~~C_o = W_o
\end{equation}
 where 
\[
J = \left[\begin{array}{ll} 0 & 1  \\ -1 & 0 \end{array} \right];
\]
e.g., see \cite{JNP1,NJD09}.  
 Furthermore, we will assume that the  quantum observer  is coupled to the quantum plant as shown in Figure \ref{F1}. 
\begin{figure}[htbp]
\begin{center}
\includegraphics[width=8cm]{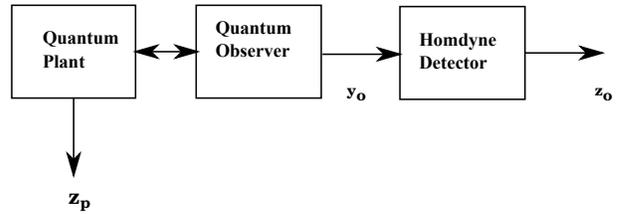}
\end{center}
\caption{Plant Observer System.}
\label{F1}
\end{figure}
 We define a coupling Hamiltonian which defines the coupling between the quantum plant and the  quantum observer:
\[
\mathcal{H}_c = x_{p}(0)\trp R_{c} x_{o}(0).
\]
The augmented quantum system consisting of the quantum plant and the quantum observer is then a  quantum system  described by the total Hamiltonian
\begin{eqnarray}
\mathcal{H}_a &=& \mathcal{H}_p+\mathcal{H}_c+\mathcal{H}_o.
\label{total_hamiltonian}
\end{eqnarray}
where the coupling operator $L$ defined in (\ref{Wo}). 

Extending the approach used in \cite{PET14Aa,PET14Ba}, we assume that $\mathcal{H}_p = 0$ and we can write
\begin{equation}
\label{Rc}
R_c = \alpha\beta^T,
\end{equation}
 $R_o = \omega_o I$, $W_o = \sqrt{\kappa} I$ where $\alpha  \in \rbb^{2}$, $\beta \in \rbb^{2}$,  $\omega_o > 0$ and $\kappa > 0$. In addition, we assume
\begin{equation}
\label{alpha}
\alpha =  C_p^T.
\end{equation}
Then, the total Hamiltonian (\ref{total_hamiltonian}) will be given by 
\[
\mathcal{H}_a = \alpha^Tx_p(0)\beta^Tx_o(0)+ \half x_o(0)\trp R_o x_o(0)
\]
since in this case the quantities $\alpha^Tx_p(0)$ and $\beta^Tx_o(0)$ are commuting scalar operators. 
Also, it follows that the augmented quantum  system is described by the equations
\begin{eqnarray}
d x_p(t) &=& -2\Theta(\alpha)x_p(t)\beta^Tx_o(t)dt;~ x_p(0)=x_{0p};\nonumber \\
d x_o(t)
&=& -\frac{\kappa}{2}x_o dt +  2 \omega_o Jx_o dt 
+2J\beta \alpha^T x_p dt - \sqrt{\kappa}dw; \nonumber \\
&& x_o(0)=x_{0o};\nonumber \\
dy_o &=& \sqrt{\kappa} x_o dt + dw; \nonumber \\
z_p(t) &=& \alpha^Tx_p(t);\nonumber \\
z_o(t) &=& Ky_o(t);
\label{augmented_system}
\end{eqnarray}
e.g., see \cite{JNP1,NJD09,EMPUJ4}.

It follows from (\ref{augmented_system}) that the quantity $z_p(t) = \alpha^Tx_p(t)$  satisfies the differential equation
\begin{eqnarray}
\label{zop}
d z_p(t)  &=& -2\alpha^T\Theta(\alpha)x_p(t)\beta^Tx_o(t)dt = 0.
\end{eqnarray}
using (\ref{eq:Theta_beta_beta}) and the fact that $\Theta(\alpha)$ is skew symmetric. That is, the quantity $z_p(t)$ remains constant and is not affected by the coupling to the coherent quantum observer:
\[
z_p(t) = z_p(0)~ \forall t \geq 0.
\]
Now using this result in (\ref{augmented_system}), it follows that 
\begin{eqnarray}
\label{xot1}
d x_o(t) &=&-\frac{\kappa}{2}x_o dt +  2 \omega_o Jx_o dt 
+2J\beta z_p dt - \sqrt{\kappa}dw.\nonumber \\
\end{eqnarray}

 Combining equations (\ref{augmented_system}), (\ref{zop}) and (\ref{xot1}), we obtain the following reduced dimension QSDEs describing the augmented quantum plant variable of interest and the quantum observer:
\begin{eqnarray}
d z_p(t)&=&0; ~ z_p(0)=\alpha^Tx_{0p};\nonumber \\
d x_o(t) &=&-\frac{\kappa}{2}x_o dt +  2 \omega_o Jx_o dt 
+2J\beta z_p dt - \sqrt{\kappa}dw; \nonumber \\
&& x_o(0)=x_{0o};\nonumber \\
dy_o &=& \sqrt{\kappa} x_o dt + dw. 
\label{augmented_system1}
\end{eqnarray}

This is a set of linear QSDEs. Hence, we can analyze this system in a similar way to \cite{PeHun2a}. To analyse the system (\ref{augmented_system1}), we first calculate the steady state value of the quantum expectation of the observer variables as follows:
\begin{eqnarray*}
<\bar x_o> &=& -2 \left[\begin{array}{ll} -\frac{\kappa}{2} & 2 \omega_o \\-2\omega_o &  -\frac{\kappa}{2} \end{array}\right]^{-1} J\beta z_p \nonumber \\
&=& \frac{4}{\kappa^2+16 \omega_o^2}\left[\begin{array}{ll} \kappa & 4 \omega_o \\-4\omega_o &  \kappa \end{array}\right]J\beta z_p.
\end{eqnarray*}
Then, we define the quantity
\[
\tilde x_o = x_o - <\bar x_o> = x_o - \frac{4}{\kappa^2+16 \omega_o^2}\left[\begin{array}{ll} \kappa & 4 \omega_o \\-4\omega_o &  \kappa \end{array}\right]J\beta z_p.
\]
We can now re-write the equations (\ref{augmented_system1}) in terms of $\tilde x_o$ as follows
\begin{eqnarray}
\label{augmented3}
d \tilde x_o &=& \left[\begin{array}{ll} -\frac{\kappa}{2} & 2 \omega_o \\-2\omega_o &  -\frac{\kappa}{2} \end{array}\right]x_o dt 
+2J\beta z_p dt - \sqrt{\kappa}dw \nonumber \\
&=&  \left[\begin{array}{ll} -\frac{\kappa}{2} & 2 \omega_o \\-2\omega_o &  -\frac{\kappa}{2} \end{array}\right]\tilde x_o dt \nonumber \\
&&- 2 \left[\begin{array}{ll} -\frac{\kappa}{2} & 2 \omega_o \\-2\omega_o &  -\frac{\kappa}{2} \end{array}\right]
\left[\begin{array}{ll} -\frac{\kappa}{2} & 2 \omega_o \\-2\omega_o &  -\frac{\kappa}{2} \end{array}\right]^{-1} J\beta z_p dt\nonumber \\
&&+2J\beta z_p dt - \sqrt{\kappa}dw \nonumber \\
&=& \left[\begin{array}{ll} -\frac{\kappa}{2} & 2 \omega_o \\-2\omega_o &  -\frac{\kappa}{2} \end{array}\right]\tilde x_o dt 
 - \sqrt{\kappa}dw; \nonumber \\
dy_o &=& \sqrt{\kappa} \tilde x_o dt - 2 \sqrt{\kappa}\left[\begin{array}{ll} -\frac{\kappa}{2} & 2 \omega_o \\-2\omega_o &  -\frac{\kappa}{2} \end{array}\right]^{-1} J\beta z_p dt+ dw \nonumber \\
&=& - 2 \sqrt{\kappa}\left[\begin{array}{ll} -\frac{\kappa}{2} & 2 \omega_o \\-2\omega_o &  -\frac{\kappa}{2} \end{array}\right]^{-1} J\beta z_p dt + dw^{out}
\end{eqnarray}
where
\[
dw^{out} = \sqrt{\kappa} \tilde x_o dt + dw.
\]

We now look at the transfer function of the system 
\begin{eqnarray}
\label{augmented4}
\dot{\tilde x}_o &=& \left[\begin{array}{ll} -\frac{\kappa}{2} & 2 \omega_o \\-2\omega_o &  -\frac{\kappa}{2} \end{array}\right]\tilde x_o  
 - \sqrt{\kappa}w; \nonumber \\
w^{out} &=& \sqrt{\kappa} \tilde x_o  + w,
\end{eqnarray}
which is given by
\[
G(s) = -\kappa \left[\begin{array}{ll}s+\frac{\kappa}{2} & - 2 \omega_o\\2\omega_o & s+\frac{\kappa}{2}\end{array}\right]^{-1}.
\]
It is straightforward to verify that this transfer function is such that
\[
G(j\omega)G(j\omega)^\dagger =I
\]
for all $\omega$. That is $G(s)$ is all pass. Also, the matrix $\left[\begin{array}{ll} -\frac{\kappa}{2} & 2 \omega_o \\-2\omega_o &  -\frac{\kappa}{2} \end{array}\right]$ is Hurwitz and hence, the system (\ref{augmented4}) will converge to a steady state in which $dw^{out}$ represents a standard quantum white noise with zero mean and  unit intensity. Hence, at steady state, the equation 
\begin{equation}
\label{yo}
dy_o = - 2 \sqrt{\kappa}\left[\begin{array}{ll} -\frac{\kappa}{2} & 2 \omega_o \\-2\omega_o &  -\frac{\kappa}{2} \end{array}\right]^{-1} J\beta z_p dt + dw^{out}
\end{equation}
shows that the output field converges to a constant value plus zero mean white quantum noise with unit intensity. 

We now consider the construction of the vector $K$ defining the observer output $z_o$. This vector determines the quadrature of the output field which is measured by the homodyne detector. We first re-write equation (\ref{yo}) as
\[
dy_o =  ez_p dt + dw^{out}
\]
where 
\begin{equation}
\label{e}
e = - 2 \sqrt{\kappa}\left[\begin{array}{ll} -\frac{\kappa}{2} & 2 \omega_o \\-2\omega_o &  -\frac{\kappa}{2} \end{array}\right]^{-1} J\beta
\end{equation}
is a vector in $\rbb^{2}$. Then
\[
dz_o = Kez_pdt + Kdw^{out}.
\]
Hence, we choose $K$ such that
\begin{equation}
\label{Kconstraint}
Ke = 1
\end{equation}
and therefore 
\[
dz_o = z_pdt + dn
\]
where 
\[
dn = K dw^{out}
\]
will be a white noise process at steady state with intensity $\|K\|^2$. Thus, to maximize the signal to noise ratio for our measurement, we wish to choose $K$ to minimize $\|K\|^2$ subject to the constraint (\ref{Kconstraint}). Note that it follows from (\ref{Kconstraint}) and the Cauchy-Schwartz inequality that
\[
1 \leq \|K\|\|e\|
\]
and hence
\[
\|K\| \geq \frac{1}{\|e\|}. 
\]
However, if we choose 
\begin{equation}
\label{K}
K = \frac{e^T}{\|e\|^2}
\end{equation}
then (\ref{Kconstraint}) is satisfied and $\|K\| = \frac{1}{\|e\|}$. Hence, this value of $K$ must be the optimal $K$. 

We now consider the special case of $\omega_o = 0$. In this case, we obtain
\[
e =  2 \sqrt{\kappa}\left[\begin{array}{ll} \frac{2}{\kappa} & 0 \\0 & \frac{2}{\kappa} \end{array}\right] J\beta = \frac{4}{\sqrt{\kappa}}J\beta.
\]
Hence, as $\kappa \rightarrow 0$, $\|e\| \rightarrow \infty$ and therefore $\|K\| \rightarrow 0$. This means that we can make the noise level on our measurement arbitrarily small by choosing $\kappa > 0$ sufficiently small. However, as $\kappa$ gets smaller, the system (\ref{augmented4}) gets closer to instability and hence, takes longer to converge to steady state. 

\section{Kalman Filter for the Plant Observer System}
Since the QSDEs (\ref{augmented_system1}) describing the plant observer system are linear, it should be possible to apply Kalman filtering to this system in order to estimate $z_p$ based on the available measurements. However, the QSDEs (\ref{augmented_system1}) are not physically realizable; e.g., see \cite{JNP1,PET10B}. Hence, the quantum Kalman filter such as discussed in \cite{YAM06,WM10} formally does not apply; see also \cite{BHJ07}. Although the QSDEs (\ref{augmented_system1}) could be made physically realizable by adding an extra fictitious quadrature variable to pair with $z_p$, using the technique described in \cite{WNZJ13},  the issue would remain that $z_p$ corresponds to a finite level quantum system and hence, its initial condition cannot be Gaussian. To overcome this issue, we will take another approach to Kalman filtering for quantum systems noting that in the classical case, the Kalman filter also has the property that it is optimal linear unbiased estimator for a linear stochastic system, even in the case of non-Gaussian noise and initial conditions; e.g., see \cite{AT67,KS72}.

We first consider a general set of linear QSDSs:
\begin{eqnarray}
dx(t) &=& A(t)x(t)dt+B(t)dw;\quad x(t_0)=x_{0};\nonumber \\
dy(t) &=& C(t)x(t)dt+dw(t)
 \label{QSDEs}
\end{eqnarray}
where $x$ is a $n \times 1$ vector of  self adjoint operators on an underlying Hilbert space, 
$dw$ is a $m\times 1$ vector of quantum noises expressed in quadrature form corresponding to the input field of the system and $dy$ represents the corresponding output field; e.g., see \cite{JNP1,NJD09,PET10B}. Here $m$ is assumed to be even. The measured output of the system $z(t)$ will be  a real vector quantity of dimension $\frac{m}{2}$ obtained by applying homodyne detection to yield one quadrature of each of the output fields; i.e., we write
\begin{equation}
\label{homodyne}
dz(t) = Ddy(t) =DC(t)x(t)dt+Ddw(t).
\end{equation}
Also, $A(t) \in \rbb^{n\times n}$, $B(t) \in \rbb^{n\times m}$, $C(t)\in \rbb^{m \times n}$, $D\in \rbb^{\frac{m}{2}\times m}$. The augmented system (\ref{augmented_system1}) is a system of the form (\ref{QSDEs}).  

We will consider linear filters of the following form:
\begin{equation}
\label{filter1}
d \hat{x}(t) = F(t)\hat{x}(t)dt + G(t) dz(t); \quad \hat{x}(t_0) = \hat{x}_{0};
\end{equation}
where $\hat{x}(t)\in \rbb^{n}$ is a vector of estimates for $x$,  $F(t) \in \rbb^{n\times n}$ and $G(t)\in \rbb^{n \times \frac{m}{2}}$.  The filter (\ref{filter1}) is said to be an {\em unbiased estimator} for the system (\ref{QSDEs}) if
\[
<x(t)> = \mathbb{E}\left\{\hat{x}(t)\right\} \quad \forall t \geq t_0;
\]
e.g., see \cite{AT67}. Here $<x(t)> = \tr(\rho x(t))$ denotes the quantum expectation of $x(t)$ where $\rho$ is the system density operator; e.g. see \cite{JNP1,NJD09,PET10B}. It is straightforward to verify  that if (\ref{filter1}) is an unbiased estimator for the system (\ref{QSDEs}) then
\[
F(t) = A(t) - G(t) DC(t) \quad \forall t \geq  t_0
\]
and $\hat{x}_0 = <x_0>$; e.g., see \cite{AT67}. Hence, an unbiased estimator for the system (\ref{QSDEs}) will be a filter of the form 
\begin{eqnarray}
\label{filter2}
d \hat{x}(t) &=& \left(A(t) - G(t) DC(t)\right)\hat{x}(t)dt + G(t) dz(t); \nonumber \\
 \hat{x}(t_0) &=& <x_0>.
\end{eqnarray}
Corresponding to the system (\ref{augmented_system1}) and the filter (\ref{filter2}) is the estimation error
\[
e(t) = x(t) - \hat{x}(t)
\]
which satisfies 
\[
de(t) = \left[A(t) - G(t)D C(t)\right]e(t)dt +\left[B(t) - G(t)D\right]dw(t).
\]
The corresponding error variance is defined by 
\[
J = <e(T)^Te(T)> = \tr[\Sigma(T)]
\]
where
\[
\Sigma(T) = \frac{1}{2}<e(t)e(t)^T+(e(T)e(T)^T)^T>
\]
is the error covariance matrix. It is straightforward to verify that the matrix $\Sigma(T)$ satisfies the following matrix differential equation:
\begin{eqnarray*}
\dot \Sigma(t) &=& [A(t) - G(t)DC(t)]\Sigma(t)\nonumber \\
&&+\Sigma(t)[A(t)-G(t)DC(t)]^T\nonumber \\
&&+\left[B(t) - G(t)D\right]\left[B(t) - G(t)D\right]^T
\end{eqnarray*}
where 
\begin{eqnarray*}
\lefteqn{\Sigma(t_0)}\nonumber \\
 &=& \frac{1}{2}<e(t_0)e(t_0)^T+(e(t_0)e(t_0)^T)^T> \nonumber \\
&=& \frac{1}{2}<(x(t_0)-<x_0>)(x(t_0)-<x_0>)^T>\nonumber \\
&&+\frac{1}{2}<((x(t_0)-<x_0>)(x(t_0)-<x_0>)^T)^T>\nonumber \\
&=& \Sigma_0;
\end{eqnarray*}
e.g., see \cite{AT67}. The filter of the form (\ref{filter2}) which minimizes the quantity $J$ is the minimum variance unbiased estimator for the system (\ref{QSDEs}). This filter is a version of the Kalman filter for the case of general QSDEs of the form (\ref{QSDEs}).

\begin{theorem}
\label{T1}
The minimum variance unbiased estimator for the system (\ref{QSDEs}) is a filter of the form (\ref{filter2}) where
\[
G(t) = \left(\Sigma^*(t)C(t)^TD^T+B(t)D^T\right)\left(DD^T\right)^{-1}
\]
and $\Sigma^*(t)$ is defined by the matrix differential equation
\begin{eqnarray*}
\dot \Sigma^*(t) &=& [A(t) - B(t)D^T\left(DD^T\right)^{-1}DC(t)]\Sigma^*(t)\nonumber \\
&&+\Sigma^*(t)[A(t)- B(t)D^T\left(DD^T\right)^{-1}DC(t)]^T\nonumber \\
&&-\Sigma^*(t)C(t)^TD^T\left(DD^T\right)^{-1}DC(t)\Sigma^*(t)\nonumber \\
&&+B(t)B(t)^T-B(t)D^T\left(DD^T\right)^{-1}DB(t)^T;\nonumber \\
\Sigma^*(t_0) &=& \Sigma_0. 
\end{eqnarray*}
Furthermore, the error covariance matrix for this estimator is given by $\Sigma(t) \equiv \Sigma^*(t)$.
\end{theorem} 

\begin{proof}
The proof of this result follows by an identical argument to the proof of the corresponding classical result; e.g., see  \cite{AT67,KS72}.
\end{proof}

We now construct the above Kalman filter for the system (\ref{augmented_system1}). We assume that the density operator for the quantum plant (\ref{plant}) is $\rho_p$.  It follows that for $t_0 = 0$, 
\[
<z_p(t_0)> = \sum_{i=1}^3C_{pi}\tr(\rho_p\sigma_i) = \bar z_{p0}
\]
and write
\[
<x_o(t_0)>  = \bar x_{o0}. 
\]
Also using (\ref{Pauli_prod}), we calculate
\begin{eqnarray*}
<(z_p(t_0)-\bar z_{p0})^2> &=&  <z_p(t_0)^2> - \bar z_{p0}^2 \nonumber \\
&=&\sum_{i=1}^3C_{pi}^2 -  \bar z_{p0}^2 \nonumber \\
&=&\sigma_{p0}
\end{eqnarray*}
and write
\begin{eqnarray*}
&&\frac{1}{2}<(x_o(t_0)-\bar x_{o0})(x_o(t_0)-\bar x_{o0})^T>\nonumber \\
&&+\frac{1}{2}<((x(t_0)-\bar x_{o0})(x(t_0)-\bar x_{o0})^T)^T>\nonumber \\
&=& \Sigma_{o0}.
\end{eqnarray*}

Now the plant observer system (\ref{augmented_system1}) defines a set of QSDEs of the form (\ref{QSDEs}), (\ref{homodyne}) where
\begin{eqnarray*}
A(t) &\equiv& \left[\begin{array}{ll}0 & 0 \\
2J\beta & -\frac{\kappa}{2}I+2\omega_oJ\end{array}\right]; \nonumber \\
B(t) &\equiv& \left[\begin{array}{l}0 \\ -\sqrt{\kappa}I\end{array}\right]; \nonumber \\
 C(t) &\equiv& \left[\begin{array}{ll}0 & \sqrt{\kappa}I\end{array}\right]; \quad D= K; \nonumber \\
<x_0> &=&\left[\begin{array}{l}\bar z_{p0} \\ \bar x_{o0} \end{array}\right]; \quad 
\Sigma_0 = \left[\begin{array}{ll} \sigma_{p0} & 0 \\ 0 & \Sigma_{o0} \end{array}\right].
\end{eqnarray*}
Hence, the corresponding Kalman filter for the plant observer system (\ref{augmented_system1}) is defined by the equations
\begin{eqnarray*}
d \left[\begin{array}{l} \hat z_p \\ \hat x_o \end{array}\right]
&=& \left[\begin{array}{ll}0 & 0 \\
2J\beta & -\frac{\kappa}{2}I+2\omega_oJ\end{array}\right]\left[\begin{array}{l} \hat z_p \\ \hat x_o \end{array}\right]dt \nonumber \\
&&+ G(t)\left(dz_o - \left[\begin{array}{ll}0 & \sqrt{\kappa}K\end{array}\right]\left[\begin{array}{l} \hat z_p \\ \hat x_o \end{array}\right]dt\right);\nonumber \\
 \left[\begin{array}{l} \hat z_p(0) \\ \hat x_o(0) \end{array}\right] &=& \left[\begin{array}{l}\bar z_{p0} \\ \bar x_{o0} \end{array}\right];\nonumber \\
 G(t) &=& \left(\Sigma^*(t)
+I\right)\left[\begin{array}{l}0 \\ \sqrt{\kappa}I\end{array}\right]K^T\left(KK^T\right)^{-1};
\end{eqnarray*}
{\small
\begin{eqnarray*}
\lefteqn{ \dot \Sigma^*}\nonumber \\
&=& \left[\begin{array}{ll}0 & 0 \\
2J\beta & -\frac{\kappa}{2}I+2\omega_oJ+\kappa K^T\left(KK^T\right)^{-1}K\end{array}\right]\Sigma^*\nonumber \\
 &+&\Sigma^*\left[\begin{array}{ll}0 & 0 \\
2J\beta & -\frac{\kappa}{2}I+2\omega_oJ+\kappa K^T\left(KK^T\right)^{-1}K\end{array}\right]^T\nonumber \\
&-& \Sigma^*\left[\begin{array}{ll}0 & 0 \\
0 & \kappa K^T\left(KK^T\right)^{-1}K\end{array}\right]\Sigma^*\nonumber \\
 &+&\left[\begin{array}{ll}0 & 0 \\
0 & \kappa I -\kappa K^T\left(KK^T\right)^{-1}K\end{array}\right];\nonumber \\
\Sigma^*(0) &=& \left[\begin{array}{ll} \sigma_{p0} & 0 \\ 0 & \Sigma_{o0} \end{array}\right].
\end{eqnarray*}
}


\begin{thebibliography}{10}
\providecommand{\url}[1]{#1}
\csname url@rmstyle\endcsname
\providecommand{\newblock}{\relax}
\providecommand{\bibinfo}[2]{#2}
\providecommand\BIBentrySTDinterwordspacing{\spaceskip=0pt\relax}
\providecommand\BIBentryALTinterwordstretchfactor{4}
\providecommand\BIBentryALTinterwordspacing{\spaceskip=\fontdimen2\font plus
\BIBentryALTinterwordstretchfactor\fontdimen3\font minus
  \fontdimen4\font\relax}
\providecommand\BIBforeignlanguage[2]{{%
\expandafter\ifx\csname l@#1\endcsname\relax
\typeout{** WARNING: IEEEtran.bst: No hyphenation pattern has been}%
\typeout{** loaded for the language `#1'. Using the pattern for}%
\typeout{** the default language instead.}%
\else
\language=\csname l@#1\endcsname
\fi
#2}}

\bibitem{VP9a}
I.~Vladimirov and I.~R. Petersen, ``Coherent quantum filtering for physically
  realizable linear quantum plants,'' in \emph{Proceedings of the 2013 European
  Control Conference}, Zurich, Switzerland, July 2013.

\bibitem{MEPUJ1a}
Z.~Miao, L.~A.~D. Espinosa, I.~R. Petersen, V.~Ugrinovskii, and M.~R. James,
  ``Coherent quantum observers for n-level quantum systems,'' in
  \emph{Australian Control Conference}, Perth, Australia, November 2013.

\bibitem{MJP1}
Z.~Miao, M.~R. James, and I.~R. Petersen, ``Coherent observers for linear
  quantum stochastic systems,'' \emph{Automatica}, vol.~71, pp. 264--271, 2016.

\bibitem{PET14Aa}
I.~R. Petersen, ``A direct coupling coherent quantum observer,'' in
  \emph{Proceedings of the 2014 IEEE Multi-conference on Systems and Control},
  Antibes, France, October 2014, also available arXiv 1408.0399.

\bibitem{PET14Ba}
------, ``A direct coupling coherent quantum observer for a single qubit finite
  level quantum system,'' in \emph{Proceedings of 2014 Australian Control
  Conference}, Canberra, Australia, November 2014, also arXiv 1409.2594.

\bibitem{PET14Ca}
------, ``Time averaged consensus in a direct coupled distributed coherent
  quantum observer,'' in \emph{Proceedings of the 2015 American Control
  Conference}, Chicago, IL, July 2015.

\bibitem{PET14Da}
------, ``Time averaged consensus in a direct coupled coherent quantum observer
  network for a single qubit finite level quantum system,'' in
  \emph{Proceedings of the 10th ASIAN CONTROL CONFERENCE 2015}, Kota Kinabalu,
  Malaysia, May 2015.

\bibitem{PeHun1a}
I.~R. Petersen and E.~H. Huntington, ``A possible implementation of a direct
  coupling coherent quantum observer,'' in \emph{Proceedings of 2015 Australian
  Control Conference}, Gold Coast, Australia, November 2015.

\bibitem{PeHun3a}
------, ``A reduced order direct coupling coherent quantum observer for a
  complex quantum plant,'' in \emph{Proceedings of the European Control
  Conference 2016}, Aalborg, Denmark, June 2016.

\bibitem{PeHun2a}
I.~R. Petersen and E.~Huntington, ``Implementation of a direct coupling
  coherent quantum observer including observer measurements,'' in
  \emph{Proceedings of the 2016 American Control Conference}, Boston, MA, July
  2016.

\bibitem{EMPUJ4}
L.~A.~D. Espinosa, Z.~Miao, I.~R. Petersen, V.~Ugrinovskii, and M.~R. James,
  ``Physical realizability and preservation of commutation and anticommutation
  relations for n-level quantum systems,'' \emph{SIAM Journal on Control and
  Optimization}, vol.~54, no.~2, pp. 632--661, 2016.

\bibitem{AT67}
M.~Athans and E.~Tse, ``A direct derivation of the optimal linear filter using
  the maximum principle,'' \emph{IEEE Transactions on Automatic Control}, vol.
  AC-12, no.~6, pp. 690--698, 1967.

\bibitem{KS72}
H.~Kwakernaak and R.~Sivan, \emph{Linear Optimal Control Systems}.\hskip 1em
  plus 0.5em minus 0.4em\relax Wiley, 1972.

\bibitem{JNP1}
M.~R. James, H.~I. Nurdin, and I.~R. Petersen, ``${H}^\infty$ control of linear
  quantum stochastic systems,'' \emph{IEEE Transactions on Automatic Control},
  vol.~53, no.~8, pp. 1787--1803, 2008.

\bibitem{PET10B}
I.~R. Petersen, ``Quantum linear systems theory,'' \emph{Open Automation and
  Control Systems Journal}, vol.~8, pp. 67--93, 2016.

\bibitem{NJD09}
H.~I. Nurdin, M.~R. James, and A.~C. Doherty, ``Network synthesis of linear
  dynamical quantum stochastic systems,'' \emph{SIAM Journal on Control and
  Optimization}, vol.~48, no.~4, pp. 2686--2718, 2009.

\bibitem{GJ09}
J.~Gough and M.~R. James, ``The series product and its application to quantum
  feedforward and feedback networks,'' \emph{IEEE Transactions on Automatic
  Control}, vol.~54, no.~11, pp. 2530--2544, 2009.

\bibitem{ZJ11}
G.~Zhang and M.~James, ``Direct and indirect couplings in coherent feedback
  control of linear quantum systems,'' \emph{IEEE Transactions on Automatic
  Control}, vol.~56, no.~7, pp. 1535--1550, 2011.

\bibitem{YAM06}
N.~Yamamoto, ``Robust observer for uncertain linear quantum systems,''
  \emph{Phys. Rev. A}, vol.~74, pp. 032\,107--1 -- 032\,107--10, 2006.

\bibitem{WM10}
H.~M. Wiseman and G.~J. Milburn, \emph{Quantum Measurement and Control}.\hskip
  1em plus 0.5em minus 0.4em\relax Cambridge University Press, 2010.

\bibitem{BHJ07}
L.~Bouten, R.~{van~Handel}, and M.~James, ``An introduction to quantum
  filtering,'' \emph{SIAM J. Control and Optimization}, vol.~46, no.~6, pp.
  2199--2241, 2007.

\bibitem{WNZJ13}
S.~Wang, H.~I. Nurdin, G.~Zhang, and M.~R. James, ``Quantum optical realization
  of classical linear stochastic systems,'' \emph{Automatica}, vol.~49, no.~10,
  pp. 3090 -- 3096, 2013.

\end{thebibliography}

\end{document}